\documentclass[11pt,reqno]{amsart}
\usepackage{amsfonts}
\usepackage{mathrsfs}
\usepackage{graphicx} 
\usepackage{epsfig}
\usepackage{amsmath, amsthm}
\usepackage{amssymb}
\usepackage{slashed}
\usepackage{xcolor}
\usepackage{float}
\allowdisplaybreaks
\date{}

\newtheorem{theorem}{Theorem}[section]
\newtheorem{lemma}[theorem]{Lemma}
\newtheorem{definition}[theorem]{Definition}

\newtheorem{example}[theorem]{Example}

\def\diag{{\rm diag \/ }}

\theoremstyle{plain}
\newtheorem{thm}{Theorem}[section]
\newtheorem{prop}[thm]{Proposition}
\newtheorem{lem}[thm]{Lemma}
\newtheorem{cor}[thm]{Corollary}

\theoremstyle{definition}

\newtheorem{rem}[thm]{Remark}
\newtheorem{defn}[thm]{Definition}
\newtheorem{eg}[thm]{Example}
\newtheorem{subtitle}[thm]{}
\newtheorem{ex}{Exercise}[section]
\numberwithin{equation}{section}

\def\l{\lambda}

\def\n{\,\vert\,}

\def\w{\omega}

\def\cb{{\mathcal{B}}}

\def\cg{{\mathcal{G}}}

\def\cj{{\mathcal{J}}}
\def\ck{{\mathcal{K}}}
\def\cl{{\mathcal{L}}}
\def\cm{{\mathcal{M}}}

\def\cp{{\mathcal{P}}}

\def\ct{{\mathcal{T}}}

\def\li{\langle}
\def\ri{\rangle}
\def\n{\ \vert\ }

\def\bs{\bigskip}
\def\ms{\medskip}

\def\ni{\noindent}
\def\ti{\tilde}
\def\p{\partial}

\def\diag{{\rm diag}}

\def\C{\mathbb{C}}

\def\R{\mathbb{R} }
\def\Z{\mathbb{Z}}

\newcommand{\beq}{\begin{equation}}
	\newcommand{\eeq}{\end{equation}}
\newcommand{\beg}{\begin{eg}}
	\newcommand{\eeg}{\end{eg}}
\newcommand{\bthm}{\begin{thm}}
	\newcommand{\ethm}{\end{thm}}
\newcommand{\bprop}{\begin{prop}}
	\newcommand{\eprop}{\end{prop}}
\newcommand{\bcor}{\begin{cor}}
	\newcommand{\ecor}{\end{cor}}
\newcommand{\blem}{\begin{lem}}
	\newcommand{\elem}{\end{lem}}
\newcommand{\bca}{\begin{cases}}
	\newcommand{\eca}{\end{cases}}
\newcommand{\brem}{\begin{rem}}
	\newcommand{\erem}{\end{rem}}
\newcommand{\bpm}{\begin{pmatrix}}
	\newcommand{\epm}{\end{pmatrix}}
\newcommand{\bbm}{\begin{bmatrix}}
	\newcommand{\ebm}{\end{bmatrix}}
\newcommand{\bvm}{\begin{vmatrix}}
	\newcommand{\evm}{\end{vmatrix}}
\newcommand{\bdefn}{\begin{defn}}
	\newcommand{\edefn}{\end{defn}}
\newcommand{\bsub}{\begin{subtitle}}
	\newcommand{\esub}{\end{subtitle}}
\newcommand{\bex}{\begin{ex}}
	\newcommand{\eex}{\end{ex}}
\newcommand{\ben}{\begin{enumerate}}
	\newcommand{\een}{\end{enumerate}}
\newcommand{\beqq}{\begin{equation*}}
	\newcommand{\eeqq}{\end{equation*}}

\def\rd{{\rm d\/}}

\begin{document}

	\title[Nonlocal DNLS equation]
	{Loop Algebra Splitting and Darboux Transformations for Nonlocal Derivative Nonlinear Schrödinger Hierarchies} \today
	
	\author{Ziqi Li$^\dag$}
	\address{School of Mathematics (Zhuhai)\\ Sun Yat-sen University\\ Zhuhai, Guangdong, 519082, China. Email: lizq88@mail2.sysu.edu.cn }
	\author{Zhiwei Wu$^*$}\thanks{$^*$\/The research of Wu, Z. is supported by the National Natural Science Foundation of China under Grant No. 12271535 and No. 12431008.}
	\address{School of Mathematics (Zhuhai)\\ Sun Yat-sen University\\ Zhuhai, Guangdong, 519082 , China. Email: wuzhiwei3@mail.sysu.edu.cn}
	
	\maketitle

	\ni {\bf Abstract.}
	We develop a unified algebraic framework for nonlocal derivative nonlinear Schrödinger (DNLS)-type hierarchies based on loop algebra splittings. Within this framework, nonlocal and reverse space-time reductions are realized through algebraic constraints, leading to the corresponding integrable hierarchies. By constructing simple elements and establishing the associated factorization theory, we derive Darboux transformations (DTs). And apply DTs  to construct explicit  solutions.

		\ni{\bf Keywords.} nonlocal DNLS-type equations, algebraic constraints, Darboux transformation.
	
	\ni {\bf MSC2020.} 17B80, 37K10, 37K30, 37K35
	
	
	\section{Introduction}
	
	The reduction of a soliton hierarchy usually comes from the constraints of its algebraic structure. For example, the NLS hierarchy can be reduced from the AKNS hierarchy, while the DNLS hierarchies can be considered as a reduction from the Kaup-Newell (KN) system \cite{KN78}, \cite{WH16}, \cite{LTW21}. Both these two types of reduction can be consider as reduced from $SL(2)$-valued flows, which is  sub-flows on $U(2)$-valued Lax pair of flows. Their rich algebraic structures provide a natural framework for constructing Darboux and Bäcklund transformations. In particular, when an integrable hierarchy is generated from a loop group splitting, the corresponding DTs can often be derived algebraically through loop group factorizations, where the generators are given by simple elements \cite{TU00}.  
	
	Nonlocal constrained reduction of the soliton hierarchies have been studied recently. In \cite{AM13}, \cite{AM16}, \cite{AM17}, Ablowitz and Musslimani studied the inverse scattering method for nonlocal NLS equations. Zhou constructed Darboux transformation for the following nonlocal DNLS equations  in \cite{Zhou18}.
	\beq\label{dnls}
	iq_t(x,t) = q_{xx}(x,t) + \varepsilon(q(x,t)^2\bar q(-x,t))_x,\quad  (\varepsilon=\pm 1).
	\eeq
	
	 In \cite{Yang}, it shows that many of these nonlocal integrable equations can be converted to their local integrable counterparts through variable transformations, such nonlocal equations include the $\cp\ct$-symmetric Davey–Stewartson equations (DS), the reverse space-time complex-modiﬁed Korteweg–de Vries equation (CMKdV), the multidimensional reverse space-time nonlocal three wave interaction equation, $\cp\ct$-symmetric NLS and nonlocal DNLS equation.

	  For example, $\cp\ct$-symmetric NLS 
	 $$i u_t( x, t) +  u_{ x x}( x,  t)+2 \sigma u^{2}( x, t)u^{*}(-x,t),\quad  (\sigma=\pm 1).$$
	 Local NLS equation  can be obtained under $x=i\hat x, t=-\hat t$. Nonlocal DNLS equation is the same as \eqref{dnls}, corrosponding local equation can be obtained under $x=i\hat x, t=-\hat t$.  Although these variable transformations establish a close relationship between local and nonlocal equations, they do not directly reveal the underlying algebraic origin of the nonlocal reductions. In contrast, the loop group approach allows one to derive the Lax pair, the hierarchy, and the DT directly from Lie algebra splittings, providing a unified algebraic framework for both local and nonlocal integrable systems.

		Moreover, compared with classical local equations, which mainly describe isolated or closed systems such as ideal optical fiber propagation,  \cite{lou20}, \cite{AA22}, nonlocal equations offer a broader framework for modeling more complicated physical phenomena.
	
	Existing studies of nonlocal DNLS equations have mainly focused on reductions of individual integrable equations and their associated solution structures. In contrast, the present work addresses the hierarchy level by deriving nonlocal DNLS systems directly from loop algebra splittings. Terng's algebraic framework provides a systematic treatment for local DNLS hierarchies \cite{TU00}. The present work extends this philosophy to nonlocal integrable hierarchies by incorporating involutive reductions into the loop algebra splitting framework. Although we restrict ourselves to the DNLS-type hierarchies in this paper, the construction is expected to be applicable to other nonlocal integrable hierarchies arising from similar algebraic reductions.

	The purpose of this paper is to develop a systematic algebraic construction of nonlocal DNLS-type hierarchies based on loop algebra splittings and to clarify the algebraic meaning of the corresponding nonlocal constraints, allowing DTs to be obtained within the same framework and derive explicit solutions of these equations.
	
	The paper is organized as follows. In Section 2, we review the construction of DNLS-type hierarchies from loop algebra splittings. Section 3 is devoted to the construction of nonlocal and reverse space-time reductions. In Section 4, we establish the corresponding simple elements and factorization theory. Based on these results, DTs are derived in Section 5. Finally, explicit solutions obtained from the DTs are presented in Section 6.
	
	\section{The DNLS-type hierarchies from loop algebra splittings}\label{sp}
	
	There is a standard process to generate Lax pairs for certain soliton hierarchies. Let $G$ be a compact Lie group, and $\cg$ its Lie algebra. Let $L(G)$ denote the group of smooth loops from $S^1$ to $G$, and $\cl(\cg)$ its Lie algebra. Let $\sigma$ be an involution of $G$, then $\rd_e \sigma$ induces an involution on $\cg$ which is complex linear, where $e$ is the identity element in $G$. Without ambiguity, we still use $\sigma$ to denote the involution on $\cg$. 
	
	 Elements in $\cl_{\sigma}(\cg)$ can be written a power series of $\l$: $$\cl_\sigma(\cg)=\{A(\l)=\sum_{i}A_i \l ^i \n \sigma(A(-\l))=A(\l)\}.$$ 
	 Let $\ck$ and $\cp$ be the eigenspace of $\sigma$ in $\cg$ of eigenvalue $1$ and $-1$ respectively,
	 then $A(\l)=\sum_{i}A_i \l^i \in \cl_\sigma(\cg)$  if and only if $A_i \in \ck, i \text{ is even}, A_i \in \cp, i  \text{ is odd}.$
	 Let $L^{+}_\sigma(G)$ and $L^{-}_\sigma(G)$ be two subgroups of $L_\sigma(G)$, and $\cl^{\pm}_\sigma(\cg)$ the corresponding Lie algebras such that
	 \ben
	 \item $L^{+}_{\sigma}(G) \cap L^{-}_{\sigma}(G)={\rm Id}$, the identity of $L(G)$. 
	 \item $\cl_\sigma(\cg)=\cl^{+}_\sigma(\cg)\oplus \cl^{-}_\sigma(\cg)$ as a direct sum of linear space.
	 \een
	 Such pair $(\cl^{+}_\sigma(\cg), \cl^{-}_\sigma(\cg))$ is called a {\it splitting} of $\cl_{\sigma}(\cg)$.  A \emph{vacuum sequence} $\cj=\{J_1, J_2, \cdots\} \in \cl^{+}_{\sigma}(\cg)$ is a sequence of commuting elements in $\cl_{\sigma}(\cg)$, where $J_i$'s are analytic functions of $J_1$ in the enveloping algebra of $\cl_{\sigma}(\cg)$. Let $\pi_+$ be the projection of $\cl_{\sigma}(\cg)$ onto $\cl^{+}_\sigma(\cg)$ with respect to the decomposition $\cl_{\sigma}(\cg)=\cl^{+}_\sigma(\cg)\oplus \cl^{-}_\sigma(\cg)$. The phase space for the evolution equations is of the form:
	 \beqq\label{fa}
	 \cm=\pi_+(g_-J_1g_-^{-1}), \quad g_- \in L^{-}_\sigma(\cg).
	 \eeqq
	\bthm (\cite{TU11}) \label{sa} Given $\xi: \R \rightarrow  C^{\infty}(\R, \cm)$, there exists unique $Q_j(\xi) \in \cl_\sigma(\cg)$  for any $j \geq 1$  such that:
	\beqq\label{q}
	\bca
	[\p_x+\xi, Q_j(\xi)]=0, \\
	Q_j(J_1)=J_j, \\
	Q_j(\xi)=M_jJ_jM_j^{-1}, \quad M_j \in L^{-}_\sigma(G).
	\eca
	\eeqq
	\ethm
	
	The \emph{$j$-th flow in the $(G, \sigma)$-hierarchy} generated by the splitting $(\cl^{+}_\sigma(\cg),$ $ \cl^{-}_\sigma(\cg))$ and vacuum sequence $\cj$ is the following evolution of $\xi$:
	$$
	\xi_{t_j}=[\p_x+\xi, (Q_j(\xi))_+].
	$$
	
	There is a natural $L_-$ action on the space of solutions.  Let $L(GL(n, \C))$ be the group of smooth loops from $S^1$ to $GL(n, \C)$, if $L$ is a subgroup of $L(GL(n, \C))$, then the rational elements in $L_-$ can be computed explicitly. Moreover, elements with one or two poles often give rise to B\"{a}cklund transformation \cite{TU00}. Hence the loop group splitting presents a natural way to construct new solutions from a given solution. Moreover, the construction is algebraic, therefore, we do not need to solve ordinary equations, which is usually the case in DT.

	
	\ms
	
	\ni {\bf Matrix derivative nonlinear Schr\"{o}dinger hierarchies} \ 
	
	Let $\cg=u(n)$, consider a complex linear involution $\sigma_{1}$ on $u(n)$:
	$$\sigma_{1}(A)=I_{k, n-k}AI_{k, n-k}^{-1}.$$
	Denote $\ck$ and $\cp$ to be the eigenspace of $\sigma_{1}$ with respect to the eigenvalue $1$ and $-1$ respectively. Then 
	\begin{align*}
		\ck=u(k)\times u(n-k), \quad \cp=\bpm 0 & q \\ -\bar{q}^t & 0\epm, q \in \C^{k \times (n-k)}.
	\end{align*}
	This $\sigma_{1}$ induces an involution on $\cl(u(n))$, 
	$$\ti\sigma_{1}(f(\l))=I_{k, n-k}f(-\l)I_{k, n-k}^{-1}. 
	$$
	Let 
	$$\cl_{\sigma_1}(u(n))=\{f(\l) \in \cl(u(n)) \mid  \ti\sigma_1(f(x, t, \l))=f(x, t, \l)\}. $$
	where $\cl_{\sigma_{1}}(u(n))$ denote the Lie algebra of $L_{\sigma_{1}}(U(n))$, and $\cl_{\sigma_{1}}$ be the subalgebra of fix points in $\cl(u(n))$ under $\ti\sigma_{1}$. Then $f(\l)=\sum_if_i\l^i \in \cl_{\sigma_{1}}$ if and only if 
	\begin{align*}
		f_i \in \bca
		\ck, \quad i=2k, \\
		\cp, \quad i=2k+1. 
		\eca
	\end{align*}
	
	Let $\cb$ be a linear map on $\ck$ such that 
	\beq\label{ab}
	\cb (A)=\mu A, \quad \mu \in \R.
	\eeq 
	Decompose $\cl_{\sigma_{1}}=\cl_{\sigma_{1}}^{\cb+}\oplus\cl_{\sigma_{1}}^{\cb-}$ as a linear subspace direct sum as follows. Given $f(\l)=\sum_jf_j\l^j \in \cl_{\sigma_{1}}$, let $f(\l)^+$ be the component in $\cl_{\sigma_{1}}^{\cb+}$: 
	$$f(\l)^+=\sum_{j > 0}f_j\l^j+f_0-\cb(f_0)=\sum_{j > 0}f_j\l^j- 2\mu f_0.$$
	Let $L_{\sigma_{1}}$, $L_{\sigma_{1}}^{\cb+}$, and $L_{\sigma_{1}}^{\cb-}$ be the Lie group of $\cl_{\sigma_{1}}$, $\cl_{\sigma_{1}}^{\cb+}$ and $\cl_{\sigma_{1}}^{\cb-}$ respectively. 	
	From a direct computation, we see that $(\cl_{\sigma_{1}}^{\cb+}, \cl_{\sigma_{1}}^{\cb-})$ is a splitting of $\cl_{\sigma_{1}}$. Now we construct family of soliton hierarchies starting from this splitting. 
	
	Given $a=iI_{k, n-k}$, from a direction, the phase space of evolution equations generated by the splitting $(\cl_{\sigma_{1}}^{\cb+}, \cl_{\sigma_{1}}^{\cb-})$ and the vacuum sequence $\cj=\{a\l^{2k}, k \in \Z_+ \}$ is
	\begin{align*} 
		\cm =(g_-a\l^2g_-^{-1}) =a\l^2+u\l+P_0(u),\quad g_{-}\in L_{\sigma_{1}}^{\cb-},
	\end{align*}
	where $u=\bpm 0 & q \\ -\bar{q}^{t} & 0 \epm, P_0(u)=\frac{(\mu-1)i}{2}  \bpm q\bar{q}^t & 0 \\ 0 & -\bar{q}^tq\epm, q \in \C^{k \times (n-k)}$.
	
	By Theorem \ref{sa}, there exist a unique $Q(u, \l)=a\l^2+\sum_iQ_i\l^i \in \cl_{\sigma_{1}}$, s.t.
	\beq\label{c}
	\bca
	[\p_x+a\l^2+u\l+P_0(u), Q(u, \l)]=0, \\
	Q(u, \l)^2=-\l^4.
	\eca
	\eeq
	Then the \emph{$j$-th twisted matrix derivative nonlinear Schr\"{o}dinger  ($\mu$-DNLS)} is 
	\beqq\label{jf}
	u_t=[\p_x+a\l^2+u\l+P_0(u), (Q(u, \l)\l^{2j-2})_+].
	\eeqq
	For example, from \eqref{c}, the first several terms of $Q(u, \l)$ is 
	\begin{align*}
		& Q_1=u, \quad  Q_0=-\frac{i}{2}\bpm q\bar{q}^t & 0 \\ 0 & -\bar{q}^tq \epm, \\
		& Q_{-1}= \bpm 0 & \frac{i}{2}q_x-\frac{\mu}{2}q \bar q^t q\\ \frac{i}{2}\bar{q}_x^t+\frac{\mu}{2}\bar q^t q \bar{q}^t & 0\epm, \\
		& Q_{-2}=\bpm A_{-2} & 0 \\ 0 & B_{-2}\epm, 
	\end{align*}
	and
	\begin{align*}
		&A_{-2}=\frac{1}{4}(q_x\bar{q}^t-q\bar{q}^t_x)+(\frac{\mu}{2}-\frac{1}{8})iq\bar{q}^tq\bar{q}^t, \\
		&B_{-2}= \frac{1}{4}(\bar{q}^t_xq-\bar{q}^tq_x)-(\frac{\mu}{2}-\frac{1}{8})i\bar q^t q \bar q^t q.
	\end{align*}
	Therefore, the second flow is 
	\beq\label{fb}
	q_t=\frac{i}{2}q_{xx}- \frac{\mu}{2}(q_x\bar q^t q +q\bar{q}^tq_x)-(\mu-\frac{1}{2})q\bar{q}^t_xq+i (\frac{\mu^{2}}{2}-\frac{3}{4}\mu+\frac{1}{4})q \bar q^t q \bar q^t q.
	\eeq
	
	\brem Let $U(k, n-k)$ be the group of automorphisms of $\C^n$ preserving the following bi-linear form: 
	$$\li X, Y\ri = \bar X^t I_{k. n-k} Y.$$
	Let $u(k, n-k)$ be its Lie algebra. Following the similar algorithm as $\mu$-DNLS hierarchy with $\cg=u(k, n-k)$, we can get another hierarchy of commuting flows. For example, the second flow is
	\beq\label{fc}
	q_t=\frac{i}{2}q_{xx}+ \frac{\mu}{2}(q_x\bar q^t q +q\bar{q}^tq_x)+(\mu-\frac{1}{2})q\bar{q}^t_xq+i  (\frac{\mu^{2}}{2}-\frac{3}{4}\mu+\frac{1}{4})q \bar q^t q \bar q^t q.
	\eeq
	\erem
	
	In the end of this example, we list some examples of vector derivative nonlinear Schr\"{o}dinger (VDNLS)-type equations ($k=n-1$).
	
	\beg  \ 
	\ben
	\item $\mu=1$, equation \eqref{fb} and \eqref{fc} becomes the standard VDNLS equation:
	\beqq\label{fb1}
	q_t=\frac{i}{2}q_{xx}\mp \frac{1}{2}(|q|^2q)_x.
	\eeqq
	\item $\mu=\frac{1}{2}$, we get the VNLSII equation:
	\beqq\label{fb2}
	q_t=\frac{i}{2}q_{xx}\mp\frac{1}{4}\bar q^t (q^2)_x.
	\eeqq
	\item $\mu=0$, it is the VNLSIII equation:
	\beqq\label{fb3}
	q_t=\frac{i}{2}q_{xx}\pm\frac{1}{2}q\bar q^t_x q+\frac{i}{4}|q|^4q.
	\eeqq
	\een
	We use the name VDNLSI, II and III to correspond with the DNLSI, II and III equations (\cite{CLL79}, \cite{GI83}).
	\eeg

	\section{Involutive constraints and nonlocal DNLS hierarchies}
	In this section, we construct nonlocal and reverse space-time  DNLS hierarchies through algebraic constraints.
	\subsection{Nonlocal DNLS-Type hierarchies}\ 
	
	Consider the following conjugate linear automorphism $\sigma_2$ on $gl(n)$: 
	\beq\label{r1}
	\sigma_2(A)=-I_{k, n-k}\bar{A}^tI_{k, n-k}^{-1}.
	\eeq
	It induces an automorphism $\ti{\sigma}_2$ (of order $4$) on $\cl(gl(n))$ as following,
	\beq\label{r2}
	\ti\sigma_2(f(x, t, \l))=-I_{k, n-k}\overline{f(-x, t, \overline{ i\varepsilon\l})}^tI_{k, n-k}^{-1},\quad \varepsilon=\pm 1 .
	\eeq
	Let 
	$$\cl_{\sigma_2}(gl(n))=\{f(\l) \in \cl(gl(n)) \mid \ti\sigma_1(f(\l))=f(\l),\ti\sigma_2(f(x, t, \l))=f(x, t, \l)\}. $$
    Then $f(x,t,\l)=\sum_{i}f_i(x, t)\l^i \in \cl_{\sigma_2}(gl(n))$ if and only if 
	\begin{align*}
		&f_i \in gl(k)\times gl(n-k), -\bar{f}_i^t(-x, t)=f_i(x, t), \quad i=0 \mod(4), \\
		&f_i=\bpm 0 & q(x, t) \\ i \varepsilon\bar{q}^t(-x, t) & 0\epm, \quad i=1\mod(4),\\
		&f_i \in gl(k)\times gl(n-k), \bar{f}_i^t(-x, t)=f_i(x, t)
		\quad i=2 \mod(4), \\
		&f_i=\bpm 0 & q(x, t) \\ -i\varepsilon\bar{q}^t(-x, t) & 0 \epm, i=3\mod(4).
	\end{align*}
	Let $\cl(gl(n))$ denote the Lie algebra of $L(GL(n))$, and $\cl_{\sigma_{2}}$ be the subalgebra of fix points in $\cl(gl(n))$ under $\ti{\sigma}_{2}$. Let $\cg_{i}$ be the eigenspace of $
	\ti\sigma_{2}$ in $C^{\infty}(\R, gl(n))$ with respect to eigenvalue $i^{j}$,
	$0 \leq j \leq 3$, then
	$$f_{4k+i}\in \cg_{i},\quad k,i \in \Z, \quad 0\leq i \leq 3.$$
	Let $\cb$ be the real-linear map defined as \eqref{ab}. Let $(\cl_{\sigma_{2}}^{\cb+}, \cl_{\sigma_{2}}^{\cb-})$ be the splitting of $\cl_{\sigma_{2}}$ such that, given $A(\l)=\sum_iA_i \l^i \in \cl_{\sigma_{2}}$, the projection of $A(\l)$ onto $\cl_{\sigma_{2}}^{\cb+}$ with respect to $\cl=\cl_{\sigma_{2}}^{\cb+}\oplus\cl_{\sigma_{2}}^{\cb-}$ is 
	$$\pi_+(A(\l))=\sum_{i > 0}A_i\l^i+A_0-\cb(A_0)=\sum_{i > 0}A_i\l^i-(\mu-1) A_0.$$
	Given $a=iI_{k, n-k}$, $b=a$, from direct computation, the phase space is of the form:
	\beqq
	a\l^2+u\l+P_0(u),  \quad u=\bpm  0 & q(x, t) \\ -i\varepsilon \bar{q}^t(-x, t)\epm,
	\eeqq
	and 
	$$P_0(u)=-\frac{(\mu-1)}{2} \varepsilon\bpm q(x,t)\bar{q}^t(-x, t) & 0 \\ 0 & -\bar{q}^t(-x, t)q(x, t)\epm.$$

	 We can solve $Q(u, \l)=a\l^2+\sum_iQ_i\l^i$ uniquely from \eqref{c}, and it can be proved that 
	$$\ti\sigma_{2}(a\lambda^2+u\lambda+P_{0})=-(a\lambda^2+u\lambda+P_{0}),$$
e.s. $a\lambda^2+u\lambda+P_{0}\notin \cl_{\sigma_{2}}$, 	so this is not a standard splitting. This is the main difference between local hierachies and nonlocal hierachies. Because the order of automorphism $\ti{\sigma}_2$ is $4$, if $\ti\sigma_{2}(Q\lambda^{2(j-1)})=Q\lambda^{2(j-1)}, j=2k, k\in \Z_{+},$
	then $Q(u,\lambda)\lambda^{2(j-1)}\in \cl_{\sigma_{2}}, j=2k, k\in \Z_{+}$. Due to 
	\begin{align*}
	\ti\sigma_{2}([\p_x+a\l^2+u\l+P_0, Q\lambda^{2(j-1)}])=0,
	\end{align*}
	thus, we only study $2k$-th hierarchy for $\cl_{\sigma_{2}}$.
	
	For example, the first several terms $Q_{1}\in\cg_{3}, Q_{0}\in \cg_{2}, Q_{-1}\in \cg_{1}, Q_{-2}\in \cg_{0}$, these are 
	\begin{align*}
		& Q_1=u, Q_0=\frac{1}{2}\varepsilon\bpm q(x, t)\bar{q}^t(-x, t) & 0 \\ 0 & -\bar{q}^t(-x, t)q(x, t) \epm, Q_{-1}=\bpm 0 & A_-1\\
		B_{-1} & 0\epm, \\
		&A_{-1}=\frac{i}{2}q_x(x, t)-i\varepsilon\frac{\mu}{2}q(x, t)\bar{q}^t(-x, t)q(x, t), \\
		&B_{-1}=-\frac{1}{2}\varepsilon\bar{q}^t_{x}(-x, t)-\frac{\mu}{2}\bar{q}^t(-x,t)q(x, t)\bar{q}^t(-x,t).\\
		&Q_{-2}=\bpm  A_{-2} & 0 \\ 0 & B_{-2}\epm, \\
		&A_{-2}=-\frac{i}{4}\varepsilon(q(x)\bar{q}^t_x(-x)-q_x(x)\bar{q}^t(-x))-i(\frac{\mu}{2}-\frac{1}{8})q(x)\bar{q}^t(-x)q(x)\bar{q}^t(-x), \\
		&B_{-2}=-\frac{i}{4}\varepsilon(\bar{q}^t(-x)q_x(x)-\bar{q}^t_x(-x)q(x))+i(\frac{\mu}{2}-\frac{1}{8})\bar{q}^t(-x)q(x)\bar{q}^t(-x)q(x).
	\end{align*}
	Then the second flow is 
	\begin{align} \label{2fa}
		q(x, t)_t & =\frac{i}{2}q_{xx}(x,t)-i\varepsilon\frac{\mu}{2}q_x(x, t)\bar{q}^t(-x, t)q(x,t) \notag \\
		&  -i\varepsilon(\mu-\frac{1}{2})q(x, t)\bar{q}^t_x(-x, t)q(x, t) -i\varepsilon\frac{\mu}{2}q(x, t)\bar{q}^t(-x, t)q_x(x, t) \notag \\
		&  + i\varepsilon(\frac{\mu^{2}}{2}-\frac{3}{4}\mu+\frac{1}{4})q(x, t)\bar{q}^t(-x,t)q(x, t)\bar{q}^t(-x, t)q(x, t).
	\end{align}
	
	In particular, when $\mu=1$, \eqref{2fa} becomes 
	\begin{align} \label{2f}
		q_t(x, t) &=\frac{i}{2}q_{xx}(x, t)-\frac{i}{2}\varepsilon q_x(x, t)\bar{q}^t(-x, t)q(x, t) \notag \\
		& \quad -\frac{i}{2}\varepsilon q(x, t)\bar{q}^t_x(-x, t)q(x, t)-\frac{i}{2}\varepsilon q(x, t)\bar{q}^t(-x, t)q_x(x, t).
	\end{align}
	It is consistent with \cite{Zhou18}.
	\brem
	It is not surprising to note that $Q_{i}$ here, because for nonlocal DNLS, its first order flow is when  $k=1$, i.e.
	\beqq\label{c}
	\bca
	[\p_x+a\l^2+u\l+P_0(u), (Q(u, \l)\l^{2})_{+}]=0, \\
	Q(u, \l)^2=-\l^4,
	\eca
	\eeqq
	notice that $\ti\sigma_{2}(Q(u,\l)\l^{2})=Q(u,\l)\l^{2} \in \cl_{\sigma_{2}} $. The next flow is when $k=2$, $\ti\sigma_{2}(Q(u,\l)\l^{6})=Q(u,\l)\l^{6} \in \cl_{\sigma_{2}}, j=2k, k\in \Z_{+}$. The reason is because that $[\partial_x, \cg_{1}] \in \cg_{3}$, hence $[\partial_x, Q_{-(4k-3)}]$ generates a flow for $u\in \cg_{3}$. 
	\erem
	\subsection{Reverse Space-time DNLS-Type hierarchies}\ 
	
	Let $\ti \sigma_2$ be the automorphism induced by $\sigma_2$ on $\cl(gl(n))$ defined as:
	\beq\label{sb}
	\ti \sigma_3(f(x, t, \l))=-I_{k, n-k}\overline{f(-x, -t, \overline {i\varepsilon \l})}^t I_{k, n-k},\quad \varepsilon=\pm 1.
	\eeq
	Now consider the following loop algebra:
	\beqq
	\cl_{RST}(gl(n))=\{f(\l) \in \cl(gl(n, \C)) \mid \ti\sigma_1(f(\l))=f(\l), \ti \sigma_3(f(x, t, \l))=f(x, t, \l) \},
	\eeqq
	
	Then $f(x,t,\l)=\sum_{i}f_i(x, t)\l^i \in \cl_{RST}(gl(n))$ if and only if 
	\begin{align*}
		&f_i \in gl(k)\times gl(n-k), -\bar{f}^t_i(-x, -t)=f_i(x, t), \quad i=0 \mod(4), \\
		&f_i=\bpm 0 & q(x, t) \\ i\varepsilon\bar{q}^t(-x, -t) & 0\epm, \quad i=1\mod(4),\\
		&f_i \in gl(k)\times gl(n-k), \bar{f}_i^t(-x, -t)=f_i(x, t), 
		\quad i=2 \mod(4), \\
		&f_i=\bpm 0 & q(x, t) \\ -i\varepsilon\bar{q}^t(-x, -t) & 0 \epm, i=3 \mod(4).
	\end{align*}
	Let $\cl_{RST}$ be the subalgebra of fix points in $\cl(gl(n))$ under $\ti{\sigma}_{1},\ti{\sigma}_{3}$. Let $\hat\cg_{i}$ be the eigenspace of $
	\ti\sigma_{3}$ in $C^{\infty}(\R, gl(n))$ with respect to eigenvalue $i^{j}$,
	$0 \leq j \leq 3$ then
	$$f_{4k+i}\in \hat \cg_{i},\quad k,i \in \Z, \quad 0\leq i \leq 3.$$
	Let $\cb$ be the linear map defined by \eqref{ab}. Let $(\cl_{RST}^ {\cb+}, \cl_{RST}^{\cb-})$ be the splitting of $\cl_{RST}$ such that, the projection of $A(\l)=\sum_iA_i \l^i \in \cl_{RST}$ onto $\cl_{RST}^{\cb+}$ with respect to $\cl=\cl_{RST}^{ \cb+}\oplus \cl_{RST}^{\cb-}$ is 
	$$\pi_+(A(\l))=\sum_{i > 1}A_i\l^i+A_0-\cb(A_0)=\sum_{i > 1}A_i\l^i-(\mu-1) A_0, \quad \mu \in \R.$$
	Given $a=iI_{k, n-k}$, $b=I_{k,n-k}$, from a direction, the phase space of evolution equations generated by the splitting $(\cl_{RST}^ {\cb+}, \cl_{RST}^{\cb-})$ and the vacuum sequence $\cj=\{a\l^{2k}, k \in \Z_+ \}$ is 
	\begin{align*} 
		\cm =(g_-a\l^2g_-^{-1}) =a\l^2+u\l+P_0(u),
	\end{align*}
	where
	\beqq
	u=\bpm 0 & iq(x,t) \\ -\varepsilon\bar{q}^{t}(-x,-t) & 0 \epm \in \hat \cg_{3} ,
	\eeqq
	and 
	$$P_0(u)=-\frac{(\mu-1)}{2} \bpm \varepsilon q(x,t)\bar{q}^t(-x, -t) & 0 \\ 0 & -\varepsilon\bar{q}^t(-x, -t)q(x, t)\epm \in \hat \cg_{2}. $$
	Solving $Q(u, \l)=b\l^2+\sum_iQ_i\l^i$ uniquely from \eqref{c} and it can be proved that 
	$$\ti\sigma_{3}(a\lambda^2+u\lambda+P_{0})=-(a\lambda^2+u\lambda+P_{0}),$$
	$$\ti\sigma_{3}(Q\lambda^{2(j-1)})=-Q\lambda^{2(j-1)},\quad j=2k,k\in \Z_{+}.$$
    e.s. $a\lambda^2+u\lambda+P_{0}$ and $Q\lambda^{2(j-1)}$ $\notin \cl_{\sigma_{3}}$. This is the main difference between local hierachies and reverse space-time hierachies. Due to 
	\begin{align*}
		\ti\sigma_{3}([\p_x+a\l^2+u\l+P_0, Q\lambda^{2(j-1)}])=0,
	\end{align*}
	thus, we  can study $2k$-th hierarchy for $\cl_{RST}$.
	we get the first several terms are 
	\begin{align*}
		& Q_1=-iu, \quad Q_0=\frac{1}{2}\bpm -i\varepsilon q(x, t)\bar{q}^t(-x, -t) & 0 \\ 0 & i\varepsilon \bar{q}^t(-x, -t)q(x,t) \epm,  \\
		&Q_{-1}=\bpm 0 & A_-1\\
		B_{-1} & 0\epm, \\
		&A_{-1}=\frac{i}{2}q_x(x, t)-\frac{\mu}{2}i\varepsilon q(x, t)\bar{q}^t(-x, -t)q(x, t), \\
		&B_{-1}=\frac{1}{2}\varepsilon\bar{q}^t_x(-x, -t)+ \frac{\mu}{2}\bar{q}^t(-x,-t)q(x, t)\bar{q}^t(-x,-t).\\
		&Q_{-2}=\bpm  A_{-2} & 0 \\ 0 & -B_{-2}\epm,\\
		&A_{-2}=-\frac{\varepsilon}{4}(q\bar{q}^t_x(-x, -t)-q_x\bar{q}^t(-x, -t)) -(\frac{\mu}{2}-\frac{1}{8})q\bar{q}^t(-x,-t)q\bar{q}^t(-x, -t), \\
		&B_{-2}=-\frac{\varepsilon}{4}(\bar{q}^t(-x, -t)q_x+\bar{q}^t_x(-x, -t)q)
		+(\frac{\mu}{2}-\frac{1}{8})\bar{q}^t(-x, -t)q\bar{q}^t(-x, -t)q.
	\end{align*}
	Then the second flow can be written down as
	\begin{align}
		iq_t(x, t)& =(A_{-1})_x-\mu\varepsilon(q(x, t)\bar q^t(-x, -t)A_{-1}+A_{-1}\bar q^t(-x, -t)q(x, t))  \notag \\
		& \quad -2\mu i(A_{-2}q(x, t)-q(x, t)B_{-2}). \label{z}
 	\end{align}
	In particular, when $\mu=1$, $n=2$ and $k=1$, \eqref{z} becomes a single equation:
	\beq\label{za}
	q_t(x, t) =\frac{1}{2}q_{xx}(x, t)-\frac{1}{2}\varepsilon (q(x, t)^2\bar q(-x, -t))_x.
	\eeq
		\brem
	Similarity for reverse space-time DNLS, its first order flow is when  $k=1$, i.e. The next flow is when $k=2$, $\ti\sigma_{3}(Q(u,\l)\l^{6})=-Q(u,\l)\l^{6}, j=2k, k\in \Z_{+}$. The reason is because that $[\partial_x, \hat\cg_{1}] \in \hat\cg_{3}$, hence $[\partial_x, Q_{-(4k-3)}]$ generates a flow for $u\in \hat\cg_{3}$. 
	\erem
	\bs
	\section{Simple element and factorization}
	The loop group approach to DTs relies on the construction of simple elements and the associated factorization theory. In this section, we construct simple elements adapted to the nonlocal and reverse space-time DNLS-type hierarchies and establish the corresponding factorization theory. These algebraic results will serve as the basis for the DTs developed in the next section.
	
	 Let $U_{2},U_{3}$ be the compact form of $G=GL(n,\mathbb{C})$ defined by the involution $\hat\sigma$,$\hat\sigma_{2}$ and $\hat \sigma$,$\hat\sigma_{3}$ of $G$ respectively,  and $\hat \sigma$, $\hat\sigma_{2}$ $\hat \sigma_{3}$ the involution on $G$ corrosponding  $\sigma$, $\sigma_{2}$ and  $\sigma_{3}$ respectively. A meromorphic map $ g : \C \rightarrow G$ is said to satisfy the $U_{2}$-nolocal condition if
	 \begin{align}
	 &\hat\sigma(g(-\lambda))=I_{k,n-k}g(-\lambda)I_{k,n-k}^{-1}=g(\lambda), \label{r1}\\ &\hat\sigma_{2}(g(x,t,\lambda))=I_{k,n-k}(\overline{g(-x,t,\overline {i\varepsilon\lambda})}^{t})^{-1}I_{k,n-k}^{-1}=g(x,t,\lambda) \label{r2}.
	 \end{align}
	 And a meromorphic map $ g : \C \rightarrow G$ is said to satisfy the $U_{3}$-reverse space-time condition if $g$ satisfy \eqref{r1} and
	 \begin{equation} \label{r3}
	  \hat\sigma_{3}(g(x,t,\lambda))=I_{k,n-k}(\overline{g(-x,-t,\overline{i\varepsilon\lambda})}^{t})^{-1}I_{k,n-k}^{-1}=g(x,t,\lambda).
	\end{equation}
	  Let $R_{\hat\sigma_{2}}(G)(R_{\hat\sigma_{3}}(G))$ denote the groups of rational maps $g:\C \rightarrow G$ satisfying the $U_{2}$-nolocal condition ($U_{3}$-reverse space-time condition) and $g(\infty)=I$. A simple element in $R_{\hat\sigma_{2}}(G)(R_{\hat\sigma_{3}}(G))$ is an element in $R_{\hat\sigma_{2}}(G)(R_{\hat\sigma_{3}}(G))$ that has a minimal number of poles. In this section, we only construct  simple element in $R_{\hat\sigma_{2}}(G)(R_{\hat\sigma_{3}}(G))$ when $\varepsilon=-1$, the case of $\varepsilon=1 $ is similar .
	  
	  	  Let $\pi$ be the projection of $\C^{n}$, and $\alpha\in\C$, then simple element is
	  	  \begin{equation*}
	  	  	k_{\alpha_{1},\alpha_{2}}(\lambda)=I_{n}+\frac{\alpha_{1}-\alpha_{2}}{\lambda-\alpha_{1}}\pi^{\perp},
	  	  	\end{equation*} 
	 	and $k_{\alpha_{1},\alpha_{2}}$ generates $R(G)$. 
	 	\begin{definition}
	 		We call a projection of $\mathbb{C}^{n}$ a nonlocal projection if $\pi(x,t)=\overline\pi^{t}(-x,t)$ and if  $\pi(x,t)=\overline\pi^{t}(-x,-t)$, it is called reverse space-time projection.
	 	\end{definition}
	  \begin{theorem}\label{se}
Let $\pi$ be the  nonlocal projection onto  $V$, and $\bar\pi^{t}=\pi$, $s\in \mathbb{R} \backslash \{0\}$. Set $\pi^{\sharp}=I_{k,n-k}\pi I^{-1}_{k,n-k}$, if $\langle V, I_{k,n-k}V \rangle=0 $ then 
\begin{equation*}
g_{(1+i)s,\pi}= k_{-(1+i)s,\pi^{\sharp}} k_{(1+i)s,\pi}=I_{n}+\frac{2(1+i)s}{\lambda-(1+i)s}\pi^{\sharp}-\frac{(1+i)s}{\lambda+(1+i)s}\pi
\end{equation*}  
satisfies the $U_{2}$-nolocal condition and $U_{3}$-reverse space-time condition.
	  \end{theorem}
	  Next we derive the factorization theory for  elements satisfied the $U_{2}$-nolocal condition and $U_{3}$-reverse space-time condition.
	  	  \begin{theorem}\label{holo}
	  Let $s\in \mathbb{R}\backslash \{0\}$ be constant, $\pi$ a  nonlocal projection projection and $f: \mathbb{C}\rightarrow G $ a meromorphic map satisfying the $U_{2}$-nolocal condition ($U_{3}$-reverse space-time condition) and being holomorphic at $\lambda=(1+i)s$, $\tilde{v}=f((1+i)s)^{-1}(v)$. Then 
	  $$\tilde{f}:=g_{(1+i)s,\pi}fg^{-1}_{(1+i)s,\tilde{\pi}}$$
	is holomorphic at $\lambda=(1+i)s$, where $\tilde{\pi}$ is  a nonlocal projection projection onto $\mathbb{C}\tilde{v}$.
	  \end{theorem}
	  Let expand $g_{(1+i)s,\tilde{\pi}}(\lambda)$ and $f^{-1}(\lambda)$ at $\lambda=-(1+i)s$  by using Laurent series, we derive
	  $$\tilde{f}^{-1}(\lambda)=g_{(1+i)s,\tilde{\pi}}f^{-1}(\lambda)g_{(1+i)s,\pi}^{-1}=\sum_{i\geqslant -1}\xi_{i}(\lambda+(1+i)s)^{i},$$
	  then $$\xi_{-1}=-2(1+i)s\pi f(-(1+i)s)\tilde{\pi}^{\perp},\quad \xi_{0}=\tilde{f}^{-1}(-(1+i)s)$$
	  By Theorem \ref{holo}, $\xi_{-1}=0$.
	  \begin{lemma}
	  	Let  $g: \mathbb{C} \rightarrow GL(n,\mathbb{C})$ be holomorphic and satisfy the $U_{2}$-nolocal condition  ($U_{3}$-reverse space-time condition), $s\in \mathbb{R}$, $\pi$ a  nonlocal projection onto $\mathbb{C}v$, $\tilde{v}=f((1+i)s)^{-1}v$, and $\tilde{\pi}$ a nonlocal projection onto $\mathbb{C}\tilde{\pi}$. Set $\tilde{f}=g_{(1+i)s,\pi}fg_{(1+i)s,\tilde{\pi}}^{-1}$, $\alpha=(1+i)s$, then   $$\tilde{f}^{-1}(-\alpha)=\tilde{\pi}f^{-1}(-\alpha)+\tilde{\pi}^{\perp}f^{-1}(-\alpha)\pi^{\perp}-2\alpha\tilde{\pi}(f^{-1})_{\lambda}(-\alpha)\pi^{\perp}.$$
	  	\end{lemma}
	  	Denote 
	  	$$H_{\hat\sigma_{2}}(G)=\{\hat\sigma(f(-\lambda))=f(\lambda), \hat\sigma_{2}(f(-x,t,i\varepsilon\lambda))=f(x,t,\lambda), f \text{ holomorphic}\},$$
	  	$$H_{\hat\sigma_{3}}(G)=\{\hat\sigma(f(-\lambda))=f(\lambda), \hat\sigma_{3}(f(-x,-t,i\varepsilon\lambda))=f(x,t,\lambda), f \text{ holomorphic}\}.$$
	  	\begin{theorem}\label{Fa}
	  		Let $U_{2}$, $U_{3}$ be defined by the fixed point set of $\hat\sigma$, $\hat\sigma_{2}$ and  $\hat\sigma$, $\hat\sigma_{3}$  on G. \\
	  		(i) Assume that f satisfies the $U_{2}$-nolocal condition $($$U_{3}$-reverse space-time condition$)$, $f_{+}\in H_{\hat{\sigma}_{2}}(G)$ $($$f_{+}\in H_{\hat{\sigma}_{3}}(G)$$)$,  $f_{-}\in R_{\hat\sigma_{2}}(G)$ $($$f_{-}\in R_{\hat\sigma_{3}}(G)$$)$ such that $f = f_{+}f_{-}$, if $f$ satisfies the $U_{2}$-nolocal condition $($$U_{3}$-reverse space-time condition$)$, then $f_{+}$, $f_{-}$ satisfy the $U_{2}$-nolocal condition  $($$U_{3}$-reverse space-time condition$)$.\\
	  		(ii) Let $f\in H_{\hat{\sigma}_{2}}$ $($$f\in H_{\hat{\sigma}_{3}}$$)$, and $g\in R_{\hat\sigma_{2}}(G)$ $(g\in R_{\hat\sigma_{3}}(G))$, then exist $\tilde{g}(G)\in  H_{\hat{\sigma}_{2}}$ $(\tilde{g}(G)\in  H_{\hat{\sigma}_{3}})$ and $\tilde{f}\in R_{\hat\sigma_{2}}(G)$ $(\tilde{f}\in R_{\hat\sigma_{3}}(G))$ such that $fg=\tilde{g}\tilde{f}$.
	  	\end{theorem}
	  	Next we obtain the formula of the factorization $fg=\tilde{g}\tilde{f}$ for a simple element $f\in R_{\hat\sigma_{2}}(G)$ ($f\in R_{\hat\sigma_{3}}(G)$) to construct DT for nonlocal and reverse space-time DNLS system.
	  	\begin{theorem} \label{Fb}
	  		Let $\pi,\pi^{\sharp}$ be  nonlocal projection onto $V$ and $I_{k,n-k}V$ respctively, $s\in \mathbb{R}\backslash 0$, assume that $\langle V, I_{k,n-k}V \rangle =0$, let $f\in H_{\hat\sigma_{2}}(G)$ $(f\in H_{\hat\sigma_{3}}(G))$. $\widetilde{V}=f^{-1}((1+i)s)V$, and $\tilde{\pi}$ a nonlocal projection $($reverse space-time projection$)$ onto $\widetilde{V}$, let $\hat f=k_{(1+i)s}fk^{-1}_{(1+i)s}$, $\widehat V=\hat{f}^{-1}(-(1+i)s)I_{k,n-k}V$, and $\theta$ a  nonlocal projection $($reverse space-time projection$)$ onto $\widehat{V}$, then 
	  		$$\widehat V=(I+2(1+i)s\tilde{\pi}f^{-1}(-(1+i)s)f_{\lambda}(-(1+i)s))I_{k,n-k}\widetilde{V},$$
	  		 $$k_{-(1+i)s,\pi^{\sharp}}k_{(1+i)s,\pi}f=\tilde{f}k_{-(1+i)s,\theta}k_{(1+i)s,\pi}$$
	  	with $\tilde{f}\in H_{\hat\sigma_{2}}$ $(\tilde{f}\in H_{\hat\sigma_{3}})$ and $k_{-(1+i)s,\theta}k_{(1+i)s,\pi}\in R_{\hat\sigma_{2}}(G)$ $( R_{\hat\sigma_{3}}(G))$.
	  		\end{theorem}
	  	\begin{proof}
	  		By Theorem \ref{holo}, $$k_{-(1+i)s,\pi^{\sharp}}k_{(1+i)s,\pi}f=k_{-(1+i)s,\pi^{\sharp}}\hat fk_{(1+i)s, \tilde{\pi}}=\tilde{f}k_{-(1+i)s,\theta} $$
	  		By Theorem \ref{se}, $k_{-(1+i)s,\pi^{\sharp}}k_{(1+i)s,\pi}\in R_{\hat\sigma_{2}}(G) (R_{\hat\sigma_{3}}(G))$. And 
	  		$$\hat{f}^{-1}=k_{(1+i)s,\tilde \pi}f^{-1}k_{(1+i)s,\pi^{\sharp}}^{-1},$$
	  		which can compute $\widehat V$ directly. 
	  	\end{proof}
	  \section{Darboux transformation}
   	  We now derive a systematic algebraic construction of DTs based on loop group factorization for the nonlocal and reverse space-time DNLS-type hierarchies. The transformations preserve the reduction conditions and provide an algebraic method for generating explicit solutions. Let
   	  \begin{align*}
   	  	& L^{+}_{\hat{\sigma},\hat\sigma_{2}}(G)=\{f\in H_{\hat\sigma_{2}}(G)| f(0)=I\}, \quad L^{+}_{\hat{\sigma},\hat\sigma_{3}}(G)=\{f\in H_{\hat\sigma_{3}}(G)| f(0)=I\}\\
   	  	& L^-_{\hat{\sigma},\hat\sigma_{2}}(G)= \{ g\in L_{\hat\sigma_{2}}(G)\n g(\infty)=g_{0}\in U_{2}\}\\
        & L^-_{\hat{\sigma},\hat\sigma_{3}}(G)= \{ g\in L_{\hat\sigma_{3}}(G)\n g(\infty)=g_{0}\in U_{3}\}.
   	  \end{align*}
	  \begin{theorem}\label{DT}
	  	Let $E(x,t,\lambda)$ be the frame of a solution $u$ of the nonlocal and reverse space-time DNLS, i.e. $E(0,0,\lambda)=I$, and $g \in  L^-_{\hat{\sigma},\hat\sigma_{2}}(G)(L^-_{\hat{\sigma},\hat\sigma_{3}}(G))$. The following holds.
	  	
	  	(i) There is a factorization of $gE(x,t,\lambda)=\tilde{E}(x,t,\lambda)\tilde{g}(x,t,\lambda)$ uniquely with 
	  	$$\tilde{E}(x,t,\lambda)\in L^{+}_{\hat\sigma_{2}}(G) (L^{+}_{\hat\sigma_{3}}(G)), \quad \tilde{g}(x,t,\lambda)\in  L^-_{\hat{\sigma},\hat\sigma_{2}}(G) ( L^-_{\hat{\sigma},\hat\sigma_{3}}(G)  ).$$
	  	
  	  	(ii) Let $$\tilde{g}=\tilde{g}_{0}(x,t)+\tilde{g}_{-1}(x,t)\lambda^{-1}+\tilde{g}_{-2}(x,t)\lambda^{-2}+\dots$$
	  	then $\tilde{g}_{0}(x,t)\in U_{2}(U_{3}) $ and
	  	$$\tilde{u}=\tilde{g_{0}}(u+[\tilde{g}^{-1}_{0}\tilde{g}_{-1},a])\tilde{g}^{-1}_{0}$$
	  	is again a solution of \eqref{2f} and \eqref{za}, and $\tilde{E}(x,t,\lambda)$ is the trivialization  of $\tilde{u}$.
	  \end{theorem}
	  \begin{proof}
	  (i) Since splitting for nonlocal and reverse space-time DNLS hierarchies are
	  $$\mathcal{L}^{+}=\{f(\l) |\sum_{j > 0}f_j\l^j\}, \mathcal{L}^{-}=\{f(\l) |\sum_{j \leqslant 0}f_j\l^j\}.$$
	  The trivialization  $E(x,t,\lambda)$ of $u$ satisfies
	  \begin{equation*}
	  	\left\{\begin{aligned}
	  &E^{-1}E_{x}=a\lambda^2+u\lambda,\\
	  &E^{-1}E_{t}=b\lambda^4+u\lambda^3+Q_{0}\lambda^2+Q_{-1}\lambda,\\
	  &E(0,0,\lambda)=I,
	  \end{aligned}\right.
	  	\end{equation*}
	  	$a\lambda^2+u\lambda, b\lambda^4+u\lambda^3+Q_{0}\lambda^2+Q_{-1}\lambda $ are zero at $\lambda=0$, then $E(x,t,0)$ is constant, so $E(x,t,0)=I$ and $E(x,t,\lambda)\in L^{+}_{\hat{\sigma_{2}}} (L^{+}_{\hat{\sigma_{3}}})$.
	  	
	  	By Theorem \ref{Fa}, we have 
	  	$$g(\lambda)E(x,t,\lambda)=F(x,t,\lambda)\hat{g}$$
	  	with $F(x,t,\lambda)\in H_{\hat\sigma_{2}}(G) (H_{\hat\sigma_{3}}(G))$, $\hat{g}\in R_{\sigma_{2}}(G) (R_{\sigma_{3}}(G))$. In order to factoring $gE \in L^{+}_{\hat\sigma_{2}}(G)\times  L^-_{\hat{\sigma},\hat\sigma_{2}}(G)) $ or $L^{+}_{\hat\sigma_{3}}(G)\times  L^-_{\hat{\sigma},\hat\sigma_{3}}(G))$, we define $\phi(x,t):=F(x,t,0)$, set 
	  	\begin{equation*}
	  		\begin{aligned}
	  			\tilde{E}(x,t,\lambda)&=F(x,t,\lambda)\phi(x,t)^{-1} \in  L^{+}_{\hat\sigma_{2}}(G) ( L^{+}_{\hat\sigma_{3}}(G)),\\
	  			\tilde{g}&=\phi(x,t)\hat{g}\in L^-_{\hat{\sigma},\hat\sigma_{2}}(G) ( L^-_{\hat{\sigma},\hat\sigma_{3}}(G)),\\
	  		\end{aligned}
	  	\end{equation*}
	  	then we finish the proof.
	  	
	  	(ii)    Since $\tilde{E}=gE\tilde{g}^{-1}$, we have
	  	$$\tilde{E}^{-1}\tilde{E}_{x}=\tilde{g}(E^{-1}E_{x})\tilde{g}^{-1}-\tilde{g}_{x}\tilde{g}^{-1}=\tilde{g}(a\lambda^2+u\lambda)\tilde{g}^{-1}-\tilde{g}_{x}\tilde{g}^{-1},$$
	  	and 	$$\tilde{E}^{-1}\tilde{E}_{t}=\tilde{g}(E^{-1}E_{t})\tilde{g}^{-1}-\tilde{g}_{t}\tilde{g}^{-1}=\tilde{g}(\lambda^{2}Q(u,\lambda))_{+}\tilde{g}^{-1}-\tilde{g}_{t}\tilde{g}^{-1}.$$
	  	The trivialization $E(x,t,\lambda)$ of $u$ is holomorphic for $\lambda \in\mathbb{C}$, and $E(x,t,\lambda) \in L^{+}_{\hat\sigma,\hat\sigma_{2}} (L^{+}_{\hat\sigma,\hat\sigma_{3}})$, which implies that $-E^{-1}E_{x} $ and $E^{-1}E_{t} $ are in $ L^{+}_{\hat\sigma,\hat\sigma_{2}}$, or $-E^{-1}E_{x} $ and $-E^{-1}E_{t} $ are in $ L^{+}_{\hat\sigma,\hat\sigma_{3}}$. Since $\tilde{g}\in L^-_{\hat{\sigma},\hat\sigma_{2}}(G) (L^-_{\hat{\sigma},\hat\sigma_{3}}(G))$, hence
	  	$$\tilde{E}^{-1}\tilde{E}_{x}=(\tilde{g}(a\lambda^2+u\lambda)\tilde{g}^{-1})_{+}=a\lambda^{2}+\tilde{u}\lambda,$$
	  	and 	$$\tilde{E}^{-1}\tilde{E}_{t}=(\tilde{g}(\lambda^{2}Q(u,\lambda))_{+}\tilde{g}^{-1})_{+}=(\tilde{g}(\lambda^{2}Q(u,\lambda))\tilde{g}^{-1})_{+}$$
	  	for some $\tilde{u}.$ we take $\tilde{g}=\tilde{g}_{0}(x,t)+\tilde{g}_{-1}(x,t)\lambda^{-1}+\tilde{g}_{-2}(x,t)\lambda^{-2}+\dots$ into $a\lambda^2+\tilde{u}\lambda=(\tilde{g}(a\lambda^2+u\lambda)\tilde{g}^{-1})_{+}$, $\tilde{u}$ can be calculated directly.
	  	 If $Q(u,\lambda)$ is the solution of $[\partial_{x}+a\lambda^{2}+u\lambda, Q(u, \lambda)]=0$ and $Q(u,\lambda)$ is conjugate to $b\lambda^2$, we claim that
	  	 $Q(\tilde{u},\lambda)=\tilde{g}Q(u,\lambda)\tilde{g}^{-1}$. Hence
	  	 $$\tilde{g}(\partial_{x}+a\lambda^2+u\lambda)\tilde{g}^{-1}=\partial_{x}+\tilde{g}(a\lambda^2+u)\tilde{g}^{-1}-\tilde{g}_{x}\tilde{g}^{-1}=\partial_{x}+a\lambda^{2}+\tilde{u}\lambda,$$
	  	 and
	  	 $$\tilde{g}[\partial_{x}+a\lambda^{2}+u\lambda, Q(u, \lambda)]\tilde{g}^{-1}=[\partial_{x}+a\lambda^{2}+\tilde{u}\lambda, \tilde{g}Q(u, \lambda)\tilde{g}^{-1}].$$
	  	 Because $Q(u,\lambda)$ is conjugate to $b\lambda^2$, so is $\tilde{g}Q({u}, \lambda)\tilde{g}^{-1}$.
	  	 $$\tilde{E}^{-1}\tilde{E}_{t}=(\tilde{g}(\lambda^2Q(u, \lambda))_{+}\tilde{g}^{-1})=((\tilde{g}\lambda^2Q(u, \lambda)\tilde{g}^{-1})_{+})=(\lambda^2Q(\tilde{u},\lambda))_{+}.$$
	  	 This proves that $\tilde{u}$ is a solution of \eqref{2f} and \eqref{za} and $\tilde{E}$ is the trivialization of $\tilde{u}$.
	  \end{proof}
	  \begin{theorem}\label{NF}
Let $U_2, U_{3}, s, V$ be as defined as in Theorem \ref{Fb}, $a=iI_{k, n-k}$. Let $E(x,t,\lambda)$ be the trivialization of $u$ of nonlacal and reverse space-time DNLS hierarchies with $E(0,0,\lambda)=I$,
$$\widetilde{V}=E(x,t,(1+i)s)^{-1}(V),$$
	$$\widehat V=(I+2(1+i)s\tilde{\pi}f^{-1}(-(1+i)s)f_{\lambda}(-(1+i)s))I_{k,n-k}\widetilde{V},$$	
and $\tilde{\pi}(x,t), \theta(x,t)$ the nonlocal projection onto $\widetilde{V}$ and $\widehat{V}$ respectively. Set
$$\phi(x,t)=(2\pi^{\sharp}-I)(2\pi-I)(2\tilde{\pi}(x,t)-I)(2\theta(x,t)-I),$$
then
$$\tilde{u}=\phi(u+2(1+i)s[a,\tilde{\pi}-\theta])\phi^{-1},$$
is a new solution of \eqref{2fa} and
$$\tilde{E}(x,t,\lambda)=k_{-(1+i)s,\pi^{\sharp}}(\lambda)k_{(1+i)s,\pi}(\lambda)E(x,t,\lambda)k^{-1}_{(1+i)s,\tilde{\pi}(x,t)}k^{-1}_{-(1+i)s,\theta(x,t)}\phi(x,t)^{-1}$$
is the trivialization of  solution $\tilde{u}$.
	  \end{theorem}
	  \begin{proof}
	  	By Theorem \ref{Fb}, we can get 
	  	$$k_{-(1+i)s,\pi^{\sharp}}k_{(1+i)s,\pi}E(x,t,\lambda)=F(x,t,\lambda)k_{-(1+i)s,\theta}k_{-(1+i)s,\tilde{\pi}},$$
	  	then 
	  	$$F(x,t,\lambda)=k_{-(1+i)s,\pi^{\sharp}}k_{(1+i)s,\pi}E(x,t,\lambda)k_{-(1+i)s,\tilde{\pi}}^{-1}k_{-(1+i)s,\theta}^{-1}$$
	  	with $F(x,t,\lambda)\in H_{\hat\sigma_{2}} (H_{\hat\sigma_{3}})$. Since $k_{\pm(1+i)s,\pi^{\sharp}}=2\pi^{\sharp}-I$,$ k_{\pm(1+i)s,\pi}=2\pi-I$, $k_{\pm(1+i)s,\theta}=2\theta-I$, and $k_{\pm(1+i)s,\tilde{\pi}}=2\tilde{\pi}-I$, $E(x,t,0)=I$, thus $\phi(x,t):=F(x,t,0)=(2\pi^{\sharp}-I)(2\pi-I)(2\tilde{\pi}(x,t)-I)(2\theta(x,t)-I).$
	  	
	  	Let 
	  	$$\tilde{E}(x,t,\lambda)=F(x,t,\lambda)\phi(x,t)^{-1},    \quad \tilde{g}=\phi(x,t)k_{-(1+i)s,\theta(x,t)}k_{(1+i)s,\tilde{\pi}(x,t)},$$
	  	we take factorization of $gE(x,t,\lambda)$ again
	  	$$k_{-(1+i)s,\pi^{\sharp}}k_{(1+i)s,\pi}E(x,t,\lambda)=\tilde{E}(x,t,\lambda)\tilde{g}(x,t,\lambda)$$ 
	  	with $\tilde{E}(x,t,\lambda)\in L^{+}_{\hat\sigma,\hat\sigma_{2}} (L^{+}_{\hat\sigma,\hat\sigma_{3}})$ and $\tilde{g}\in L^{-}_{\hat\sigma,\hat\sigma_{2}} (L^{-}_{\hat\sigma,\hat\sigma_{3}})$, Obviously, $g=k_{-(1+i)s,\pi^{\sharp}}$ $k_{(1+i)s,\pi}\in L^{-}_{\hat\sigma,\hat\sigma_{2}} (L^{-}_{\hat\sigma,\hat\sigma_{3}})$. So we finish the factorization on $L^{+}_{\hat\sigma,\hat\sigma_{2}}\times L^{-}_{\hat\sigma,\hat\sigma_{2}}$ or $L^{+}_{\hat\sigma,\hat\sigma_{3}}\times L^{-}_{\hat\sigma,\hat\sigma_{3}}$.
	  	
	  	Expand 
	  	$$\tilde{g}=\tilde{g}_{0}(x,t)+\tilde{g}_{-1}(x,t)\lambda^{-1}+\tilde{g}_{-2}(x,t)\lambda^{-2}+\dots$$ as a power series at $\lambda=\infty$, and similarly
	  	$$k_{-(1+i)s,\theta}k_{-(1+i)s,\tilde{\pi}}=I+2(1+i)s(\theta-\tilde{\pi})\lambda^{-1}+\cdots,$$
	  	hence $$\tilde{g}=\phi(x,t)k_{-(1+i)s,\theta}k_{-(1+i)s,\tilde{\pi}}=\phi(x,t)+2(1+i)s\phi(x,t)(\theta-\tilde{\pi})\lambda^{-1}+\cdots.$$
	  	Compare with coefficients, we derive
	  	$$\tilde{g}=\phi(x,t), \quad \tilde{g}_{-1}=2(1+i)s\phi(x,t)(\theta-\tilde{\pi}).$$
	  	In terms of Theorem \ref{DT},
	  	$$\tilde{u}=\phi(u+2(1+i)s[a,\tilde{\pi}-\theta])\phi^{-1}$$
	  	is a soution of \eqref{2f} and \eqref{za}, $\tilde{E}$ is the trivialization of $\tilde{u}$.
	  \end{proof}
	  The preceding theory illustrates the DTs for the general nonlocal and reverse space-time DNLS hierarchies. We now turn to the corresponding reduced 2×2 DNLS case. The DT is different from the general's.
	  \blem\label{za} Let $E(x, t, \l)$ be the frame of a solution $u$ of the nolocal DNLS or reverse space-time DNLS equations. Given $v \in \C^2$ such that $\bar v ^t I_{1, 1}v=0$. Let $\pi$ be  the nonlocal projection or reverse space-time projection of $\C^2$ on to $\C v$. Set $\ti v=E(x, t, (1+i)s)^{-1}v, s \in \R$, and $\ti \pi$ the  nonlocal projection  or reverse projection of $\C^2$ onto $\C \ti v$, then $F(x, t, \l)=k_{(1+i)s, \pi}E(x, t, \l)k_{(1+i)s, \ti \pi}^{-1} \in H_{\hat\sigma_{2}} (H_{\hat\sigma_{3}}) $.
	  \elem 
	  \begin{proof} Denote $\pi^\sharp=I_{1, 1}\pi I_{1, 1}$. Since $v ^t I_{1, 1}v=0$ and the space is $\C^2$, we have $\pi^\sharp=\pi^\perp$. From \eqref{r1}, \eqref{r2} or \eqref{r3} we have 
	  	\begin{align*}
	  		\bar{\ti v}^t(-x,t)I_{1, 1}\ti v(x,t) & =\bar v^t \overline{E(-x, t, \overline{(1+i)s})^{-1}}^tI_{1, 1}E(x,t,(1+i)s)^{-1}v \\
	  		&=\bar v^t E(-(1+i)s)I_{1, 1}E((1+i)s)^{-1}v\\
	  		&=\bar v^t I_{1, 1}E((1+i)s)E((1+i)s)^{-1}v=0,
	  	\end{align*}
	  	and similarly $\bar{\ti v}^t(-x,-t)I_{1, 1}\ti v(x,t)=0 $. Hence $\ti \pi^{\sharp}=\ti \pi^\perp$.
	  	
	  	 Next, we show that $F(x, t, \l)$ satisfies the second condition of \eqref{r2} or \eqref{r3}. It suffices to prove that $I_{1, 1}F(x, t, -\l)I_{1, 1}=F(x, t, \l)$ and
	  	  $$I_{1, 1}(\overline{F(-x, t, \overline{i\l})}^{t})^{-1}I_{1, 1}=F(x, t, \l),  I_{1, 1}(\overline{F(-x, -t, \overline {i\l})}^{t})^{-1}I_{1, 1}=F(x, t, \l) $$
	  	  where $F(x, t, \l)=(I+\frac{2is}{\l -is}\pi^{\perp})E(x, t, \l)(I-\frac{2is}{\l + is}\ti \pi^{\perp})$. Indeed,
	  	\begin{align*}
	  		I_{1, 1}F(x, t, -\l)I_{1, 1}& =I_{1,1}(I-\frac{2is}{\l+is}\pi^{\perp})E(-\l)(I+\frac{2is}{\l -is}\ti \pi^{\perp})I_{1, 1} \\
	  		& =(I-\frac{2is}{\l+is}\pi)E(\l)(I+\frac{2is}{\l -is}\ti \pi)
	  	\end{align*}
	  	On the other hand,  
	  	\begin{align*}
	  		& (I+\frac{2is}{\l -is}\pi^{\perp})^{-1}(I-\frac{2is}{\l+is}\pi)= I-\frac{2is}{\l+is}, \\
	  		& (I+\frac{2is}{\l -is}\ti \pi)(I-\frac{2is}{\l + is}\ti \pi^{\perp})^{-1}=I+\frac{2is}{\l-is}, \\
	  		& (I-\frac{2is}{\l+is})(I+\frac{2is}{\l-is})=I.
	  	\end{align*}
	  	Therefore, $I_{1, 1}F(x, t, -\l)I_{1, 1}=F(x, t, \l)$. Similarly, we can prove $F(x,t,\lambda)$ satisfy the nolocal DNLS or reverse space-time DNLS condition.
	  \end{proof}
	  \bprop Let $u$, $E(x, t, \l), F(x, t, \l)$, $k_{(1+i)s, \pi}$ and $k_{(1+i)s, \ti \pi}$ as in Lemma \ref{za}, and $F_0(x, t)=F(x, t, 0) \in L^{+}_{\hat\sigma_{2}}(G) (L^{+}_{\hat\sigma_{3}}(G))$. Then $\ti E(x, t, \l)=F(x, t, \l)F_0^{-1}$ is a frame of a solution $\ti u$ of the DNLS equation, and the new solution $$\ti u=F_0(x, t)(u+2(1+i)s[a, \ti \pi])F_0(x, t)^{-1},$$
	  where $F_{0}=(2\pi-I)(2\tilde\pi-I)$.
	  
	  \eprop
	  
	  \section{Explicit soliton solution}
	  In this section, we apply the DTs established in the previous section to construct explicit soliton solutions from the vacuum solution for the nonlocal and reverse space-time DNLS hierarchies. The following examples demonstrate how the algebraic construction developed in this paper can be applied to both the general DTs and their $2\times2$ reductions'.
	  
	  \subsection{General Darboux transformations}
	  \begin{example}[Explicit solution of nonlocal DNLS hierarchy]
	  	We know $a=\bpm iI_k & 0 \\ 0 & -iI_{n-k}\epm$, the trivialization of the vacuum solution $u=0$ of the second nolocal DNLS with $E(0,0,\lambda)=I_{n}$ is
	  	\begin{align*}
	  		E(x, t, \lambda)=\bpm e^{i(\lambda^{2}x+\lambda^{4}t)} & 0 \\ 0 & e^{-i(\lambda^{2}x+\lambda^{4}t)} \epm.
	  	\end{align*}
	  	Let $\lambda=(1+i)s$, set $A=i(\lambda^2x+\lambda^4 t)=-2s^2x-4is^4 t$ and choose $V=(\omega_{k}, \omega_{n-k})^{t}\in \mathbb{C}^{n}$ be a constant vector such that $\bar\omega_{k}^{t}\omega_{k}=1, 
	  	\bar\omega_{n-k}^{t}\omega_{n-k}=1$, where $\omega_{k}, \omega_{n-k}$ are $k$ and $n-k$ dimension column vectors.
	  	 
	  	 Let $\pi$ be the nonlocal projection onto $\mathbb{C}V$, we apply DT for the second nonlocal DNLS hierarchy. A direct computation shows that
	  	 \begin{equation*}
	  	 	\begin{aligned}
	  	 \widetilde{V}&=E^{-1}(x,t,(1+i)s)V=\bpm e^{-A}\omega_{k} \\ e^A\omega_{n-k}\epm ,\\
	  	 \widehat{V}&=I_{k,n-k}\widetilde{V}(x,t)+(32is^4t+8s^2x)\tilde{\pi}\widetilde{V}.
	  	 \end{aligned}
	  	 \end{equation*}
	     We can also calculate $\tilde{\pi}(x,t)$ and $\theta(x,t)$ which are the nonlocal projection onto $\widetilde{V}$ and $\widehat{V}$, they have the form of 
	      \begin{equation*}
	     	\begin{aligned}
	     		\tilde{\pi}(x,t)=\frac{\widetilde V(x,t)\overline{\widetilde V}^{t}(-x,t)}{\overline{\widetilde V}^{t}(-x,t)\widetilde V(x,t)}=\frac{1}{2}\bpm \omega_{k}\bar\omega_{k}^{t} & e^{-2A}\omega_{k}\bar\omega_{n-k}^{t}\\
	     		e^{2A}\omega_{n-k}\bar\omega_{k}^{t} & \omega_{n-k}\bar\omega_{n-k}^{t} \epm,
	     	\end{aligned}
	     \end{equation*}
	   \begin{equation*}
	   	\begin{aligned}
	   		\theta(x,t)&=\frac{\widehat V(x,t)\overline{\widehat V}^{t}(-x,t)}{\overline{\widehat V}^{t}(-x,t)\widehat V(x,t)}=\frac{1}{1-r^{2}}(I_{k,n-k}\tilde{\pi}I_{k,n-k}+r\tilde{\pi}I_{k,n-k}-rI_{k,n-k}\tilde{\pi}-r^2\tilde{\pi})\\
	   		&=\frac{1}{2(1-r^2)}\bpm (1-r^2)\omega_{k}\bar\omega_{k}^{t} & (-1-r^2-2r)e^{-2A}\omega_{k}\bar\omega_{n-k}^{t}\\
	   		(-1-r^2+ 2r )e^{2A}\omega_{n-k}\bar\omega_{k}^{t} & (1-r^2)\omega_{n-k}\bar\omega_{n-k}^{t} \epm,
	   		\end{aligned}
   		\end{equation*}
	      where $r=32is^4t+8s^2x$, thus $[a, \tilde{\pi}-\theta]$ is
	      $$\bpm 0 & \frac{2i}{1-r}e^{-2A}\omega_{k}\bar\omega_{n-k}^{t} \\ \frac{-2i}{1+r}e^{2A}\omega_{n-k}\bar\omega_{k}^{t} & 0 \epm.$$
	      Due to $\phi(x,t)=(2\pi^{\sharp}-I)(2\pi-I)(2\tilde{\pi}(x,t)-I)(2\theta(x,t)-I)$,
	      $$\pi^{\sharp}=\frac{1}{2}\bpm \omega_{k}\bar\omega_{k}^{t} & -\omega_{k}\bar\omega_{n-k}^{t} \\ -\omega_{n-k}\bar\omega_{k}^{t} & \omega_{n-k}\bar\omega_{n-k}^{t} \epm, \quad \pi=\frac{1}{2}\bpm \omega_{k}\bar\omega_{k}^{t} & \omega_{k}\bar\omega_{n-k}^{t} \\ \omega_{n-k}\bar\omega_{k}^{t} & \omega_{n-k}\bar\omega_{n-k}^{t} \epm,$$
	      so we can get 
	      \begin{equation*}
	      	\begin{aligned}\phi(x,t)&=I-\frac{2}{1-r^2}[r(\pi+\pi^{\sharp})I_{k,n-k}-r^2(\pi+\pi^{\sharp})I_{2n}]\\
	      &=\bpm (1-\frac{2r}{1+r})\omega_{k}\bar\omega_{k}^{t} & 0 \\0 & (1-\frac{2r}{1-r})\omega_{n-k}\bar\omega_{n-k}^{t} \epm,	
	      \end{aligned}
      \end{equation*} 
       \begin{equation*}
      	\begin{aligned}\phi(x,t)^{-1}&=I+\frac{2}{1-r^2}[r(\pi+\pi^{\sharp})I_{k,n-k}+r^2(\pi+\pi^{\sharp})I_{2n}]\\
      		&=\bpm  (1+\frac{2r}{1-r})\omega_{k}\bar\omega_{k}^{t} & 0 \\ 0 &  (1-\frac{2r}{1+r})\omega_{n-k}\bar\omega_{n-k}^{t} \epm.
      	\end{aligned}
      \end{equation*} 
      From Theorem \ref{NF}, the new solution of \eqref{2f} is 
      $$\tilde{u}=2(1+i)s\phi(x,t)[a,\tilde{\pi}(x,t)-\theta(x,t)]\phi(x,t)^{-1}=\bpm 0 & \tilde{q}(x,t) \\ -i\tilde{\bar q}^{t}(x,t) & 0 \epm,$$
      where
      $$\tilde{q}(x,t)=\frac{4(1-i)s}{(1-r^2)^2}(r^3+3r-3r^2-1)e^{-2A}\omega_{k}\bar\omega_{n-k}^{t}.$$
      
      Now let $s=1,\w=\begin{pmatrix}1\\i \end{pmatrix}$, the corresponding solution is shown in Figure 1, it is a broken solution. It has two isolated points, which are $(x,t)=(\frac{1}{8},0)$, and $(x,t)=(-\frac{1}{8},0)$. Nonlocal NLS equation  also has similar solution \cite{Li}. They have certain practical physical application.
       \begin{figure}[htp]
      	\includegraphics[height=0.38\textheight]{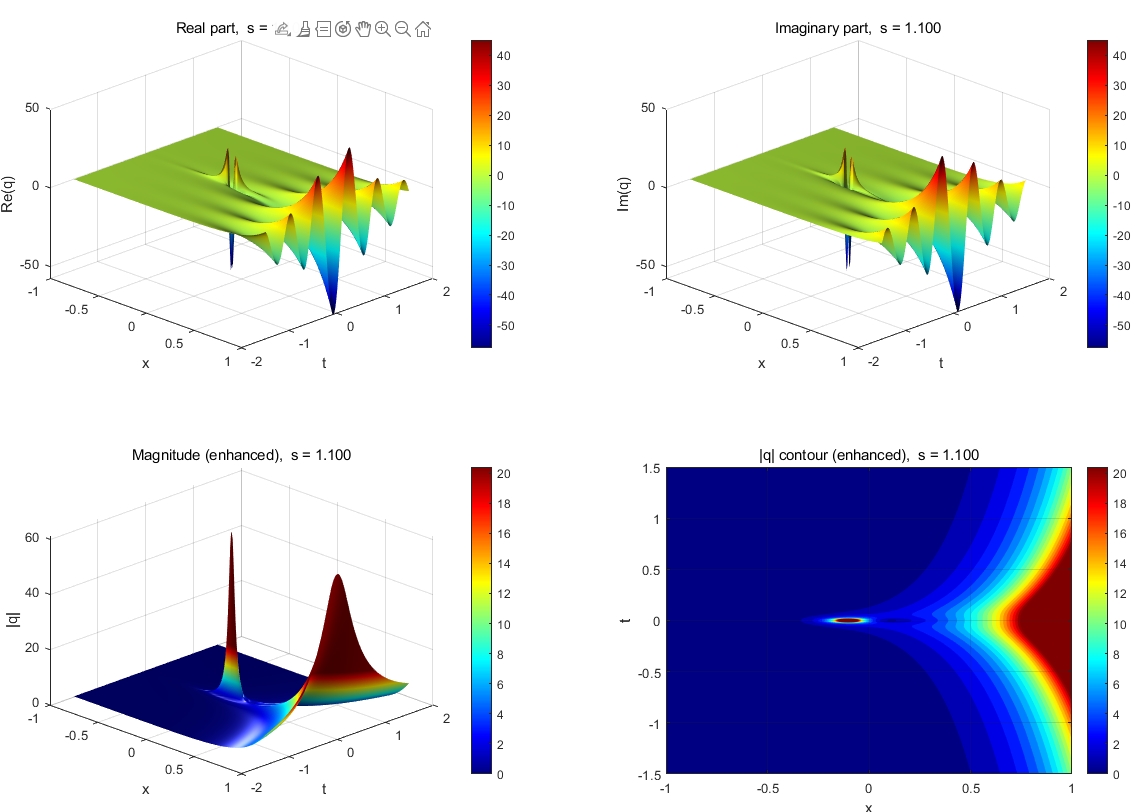}
      	\caption{Broken solution of nonlocal DNLS equation generated by the general DT.}
      \end{figure}
	  	\end{example}
	  	 \begin{example}[Explicit solution of reverse space-time DNLS hierarchy]
	  		The computation procedure of  explicit solution for reverse space-time DNLS hierarchy is similar to nonlocal DNLS hierarchy. we define a projection is a reverse space-time projection if $\pi(x,t)=\bar\pi^{t}(-x,-t)$. Repeat the above process, and the reverse space-time projection $\tilde{\pi}$ onto $\tilde{V}$ has the form 
	  		$$\tilde{\pi}=\frac{\widetilde V(x,t)\overline{\widetilde V}^{t}(-x,-t)}{\overline{\widetilde V}^{t}(-x,-t)\widetilde V(x,t)} ,$$ 
	  		and $ A=-2sx^2-4s^4t, r=8s^2x+32s^4t$, the form of $\tilde{q}$ is hold.
	  		Now let $s=1,\w=\begin{pmatrix}1\\1 \end{pmatrix}$, the corresponding solution is shown in Figure 2.  The singularities form a straight line which is $x=-\frac{1+32t}{8}$ in the $(x,t)$-plane. 
	  		\begin{figure}[htp]
	  			\includegraphics[height=0.38\textheight]{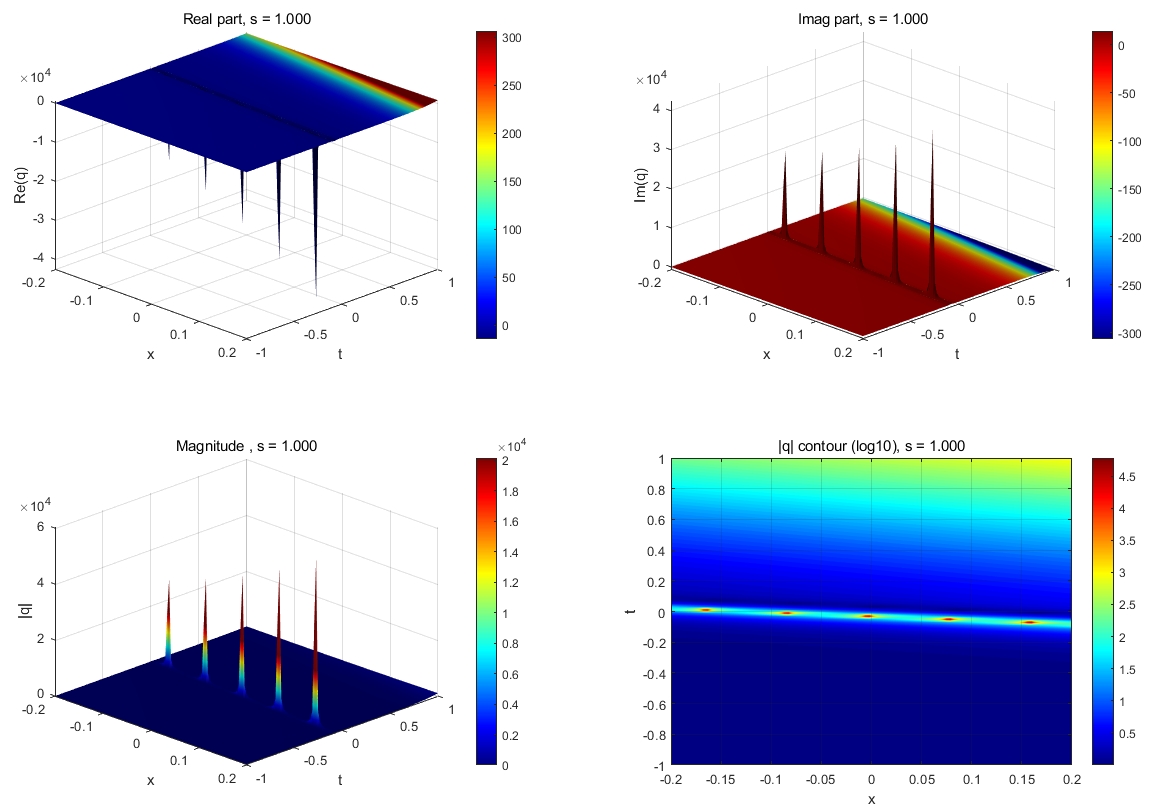}
	  			\caption{Solution of reverse space-time DNLS equation generated by the general DT.}
	  		\end{figure}
  		\end{example}
  	\subsection{Reduced Darboux transformations}
	  		 \begin{example} [Explicit solutions of the 2$\times$2 nonlocal DNLS equation]\
	  		
	  		The vacuum frame of $u=0$ is $\exp(a\l^2x+a\l^4t)$. Choose $v=(1, i)^t$, then one has 
	  		\begin{align*}
	  			& \ti v=E(x, t, (1+i)s)^{-1}v=(e^{2s^2(2is^2t + x)}, ie^{-2s^2(2is^2t + x)})^t,  \\
	  			& \ti \pi=\frac{1}{2}\bpm 1 & -ie^{4s^2(2is^2t + x)} \\ ie^{-4s^2(2is^2t + x)} & 1\epm, \\
	  			& F_0(x, t)=\diag(e^{-4s^2(2is^2t + x)}, e^{4s^2(2is^2t + x)}), \\
	  			& \ti u=\bpm 0 & (2 + 2i)se^{-4s^2(2is^2t + x)} \\ (2 + 2i)se^{4s^2(2is^2t + x)} & 0 \epm.
	  		\end{align*}
	  		Then $\ti q =(2 + 2i)se^{-4s^2(2is^2t + x)}$ is a solution of the DNLS \eqref{2f}.
	  		
	  		\end{example}
	  		 \begin{example} [Explicit solutions of the 2$\times$2 reverse-time DNLS equation]\ 
	  		
	  		The vacuum frame of $u=0$ is $\exp(a\l^2x+b\l^4t)$. Choose $v_1=(1, i)^t$, simple pole $(1+i)s_1$ then
	  		 $$\ti u=\bpm 0 & (2 + 2i)s_1e^{-8s_1^4t - 4s_1^2x} \\ (2 + 2i)s_1e^{8s_{1}^{4t}+4s_{1}^{2x}} & 0 \epm$$ is a solution of the DNLS \eqref{2f} and a new frame $\ti E(x,t,\lambda)$ of $\tilde u$ is obtained.
	  		Moreover, we choose $v_2=(1, 1)^t$, simple pole $(1+i)s_2$  to do Darboux transformation, note $z_1=s_1-s_2, z_2=s_1+s_2$
	  		$$\ti{\ti u}=\frac{-2e^{16s_1^{4}t + 8s_1^{2}x}z_1z_2(z_1+iz_{2})(s_1e^{-8z_1z_2(s_1^{2}t + s_2^{2}t + \frac{1}{2}x)}(s_1 + s_2i) - s_2s_1i - s_2^{2})}{((s_1i + s_2)s_1e^{x(4s_1^2 - 2s_2^2) + 8s_1^{4}t - 4s_4^{4}t} - e^{4s_2^{4}t + 2s_2^{2}x}s_2(s_1 + s_2i))^2}.$$
	  		
	  		Set $s_{1}=1,s_{2}=2$, 
	  		$$\ti{\ti u}=\frac{(18 + 6i)e^{16t + 8x}((-2 + i)e^{120t + 12x} + 2 - 4i)}{((2 + i)e^{56t - 4x} - (2 + 4i)e^{64t + 8x})^2},$$
	  		the corresponding solution is shown in Figure 3, it is a soliton solution.
	  			\begin{figure}[hbp]
	  			\includegraphics[height=0.38\textheight]{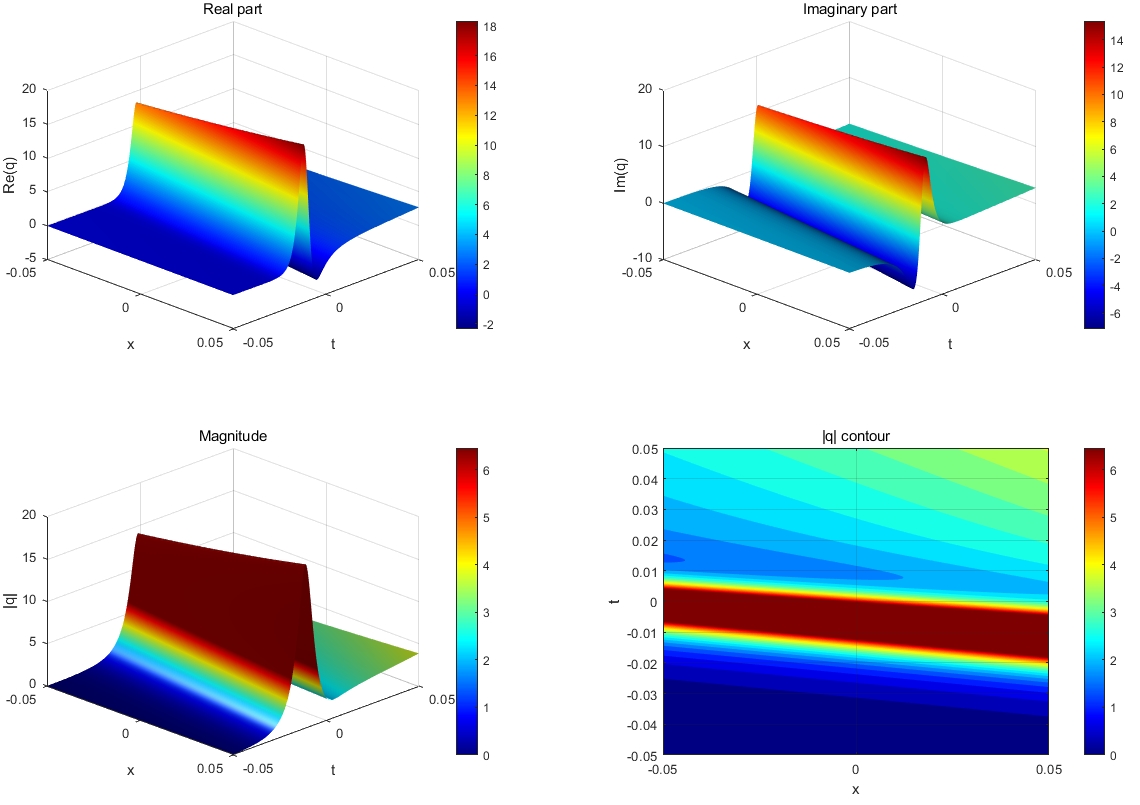}
	  			\caption{Soliton solution of 2$\times$2 reverse space-time DNLS equation  generated by the reduced DT. }
	  		\end{figure}
	  	\end{example}
  	\section{Discussion}
  	In this paper, we have developed an algebraic framework for nonlocal DNLS-type hierarchies based on Lie algebra splittings. Within this framework, nonlocal and reverse space-time nonlocal reductions are realized through algebraic constraints, leading to the corresponding integrable hierarchies. By constructing simple elements and establishing the associated factorization theory, we derive DTs for both the general hierarchies and their reduced (2$\times$2) DTs. As applications of the proposed framework, explicit soliton solutions are obtained from the vacuum solution for the nonlocal and reverse space-time DNLS equations.
  	
  	The present work demonstrates that the loop algebra splitting framework provides a unified algebraic approach to the construction of nonlocal and reverse space-time DNLS hierarchies, their DTs and explicit solutions. This viewpoint provides a natural algebraic interpretation of the corresponding nonlocal integrable hierarchies and suggests that further geometric formulation of these reductions leads to a deeper understanding of their integrable structures. Although the present paper focuses on DNLS-type hierarchies, the algebraic construction developed here relies primarily on the interplay between loop algebra splittings and involutive reductions, and is therefore expected to be adaptable to other classes of nonlocal integrable hierarchies.
  	
  	\bs
	
	
\end{document}